\documentclass[journal,twoside,web]{ieeecolor}
\usepackage{lcsys}
\usepackage{textcomp}

\usepackage{amsmath}
\usepackage{amssymb}
\usepackage[short, nocomma]{optidef}
\usepackage{algorithm}
\usepackage{algpseudocode}
\usepackage{booktabs}
\usepackage{siunitx}
\usepackage{hyperref}
\usepackage{mathtools}
\usepackage{subcaption}
\hypersetup{
    colorlinks=true,
    linkcolor=black,
    filecolor=magenta,
    urlcolor=blue,
    citecolor=black,
}

\newtheorem{theorem}{Theorem}
\newtheorem{lemma}{Lemma}
\newtheorem{corollary}{Corollary}
\newtheorem{definition}{Definition}
\newtheorem{example}{Example}
\newtheorem{remark}{Remark}

\newcommand{\new}[1]{{#1}}

\newcommand{\R}{\mathbb{R}} 

\title{
Polynomial Constraints for Robustness Analysis of Nonlinear Systems
}

\author{Neelay Junnarkar$^{1}$, \IEEEmembership{Student Member, IEEE}, Peter Seiler$^{2}$, \IEEEmembership{Fellow, IEEE}, and Murat Arcak$^{1}$, \IEEEmembership{Fellow, IEEE}
\thanks{This work was in supported in part by the National Science Foundation award CNS-2111688 and Air Force Office of Scientific Research grant FA9550-21-1-0288 and FA9550-23-1-0732.}
\thanks{$^{1}$Neelay Junnarkar and Murat Arcak are with the department
of Electrical Engineering and Computer Sciences at the University of California, Berkeley. (emails: \{neelay.junnarkar, arcak\} @berkeley.edu)}
\thanks{
$^{2}$Peter Seiler is with the Department of Electrical Engineering and Computer Science, University of Michigan, Ann Arbor, MI 48109 USA. (email: \{pseiler@umich.edu\})
}
}

\begin{document}

\maketitle
\thispagestyle{empty}
\pagestyle{empty}

\begin{abstract}
	This paper presents a framework for abstracting uncertain or non-polynomial components of dynamical systems using polynomial constraints.
	This enables the application of polynomial-based analysis tools, such as sum-of-squares programming, to a broader class of non-polynomial systems.
	A numerical method for constructing these constraints is proposed.
	The relationship between polynomial constraints and existing integral quadratic constraints (IQCs) is investigated, providing transformations of IQCs into polynomial constraints.
	The effectiveness of polynomial constraints in characterizing nonlinearities is validated via numerical examples to compute inner estimates of the region of attraction for two systems.
\end{abstract}


\section{INTRODUCTION}

\IEEEPARstart{A} classic problem in nonlinear control is the problem of absolute stability: analyzing stability of a feedback loop of a linear time-invariant (LTI) system and a nonlinearity~\cite{KOKOTOVIC2001637}.
The traditional problem considers conditions under which stability can be guaranteed for any memoryless sector-bounded nonlinearity.
Integral quadratic constraints (IQCs) are a widely used tool that generalizes this type of analysis to many classes of nonlinearities and uncertainties~ \cite{megretskiSystemAnalysisIntegral1997},
offering computationally tractable analysis methods.
Libraries of IQCs are available for various classes of nonlinearities, such as sector-bounded and slope-restricted nonlinearities.

Despite their utility, quadratic constraints can result in conservative characterizations of nonlinearities.
For example, a sector bound poorly captures the asymptotic behavior of a saturating function such as \(\tanh\).
Tight bounds of nonlinearities, and therefore a tight characterization of the nonlinear system, lead to improved performance certificates, such as larger inner estimates of a region of attraction (ROA) or smaller  \(L_2\)-gain bounds.
Motivated by this, we propose a generalization of IQCs to higher-order polynomial constraints, enabling tighter characterizations of nonlinearities.
We present definitions of polynomial constraints, a numerical method to directly synthesize a class of polynomial constraints, and establish results on deriving them through polynomial transformation of IQCs.
Finally, using sum-of-squares (SOS) programming, we demonstrate the utility of polynomial constraints by verifying inner estimates of the regions of attraction for two systems where sector constraints result in conservative ROAs.

SOS programming enables optimization with constraints sufficient for polynomial nonnegativity, making SOS programming a useful computational tool in the analysis of polynomial dynamical systems~\cite{papachristodoulou_tutorial_2005}.
This framework has been extended to non-polynomial systems by characterizing the non-polynomial parts with IQCs~\cite{papachristodoulou_tutorial_2005, iannelli_roa_iqc}.
\new{Non-polynomial systems have also been analyzed by bounding remainders of truncated Taylor expansions of non-polynomial terms within a polytope~\cite{chesi_domain_2009}; the resulting interval bounds on the non-polynomial output are a special case of the polynomial constraints $p(v,w) \geq 0$ proposed in this paper.}
Quadratic constraints themselves have been extended, taking advantage of the capabilities of SOS programming, to non-homogenous quadratics and sector constraints which do not pass through the origin~\cite{newton_neural_2021}.
Changes of variables that lift to higher dimensional spaces have been used to apply SOS techniques to large classes of nonlinearities~\cite{papachristodoulou_analysis_2005}.

Section~\ref{sec:poly-constraints} introduces polynomial constraints.
\new{Section~\ref{sec:constructing-constraints} presents methods for constructing constraints.
	\new{Section~\ref{sec:transformations} presents results on transformations of constraints, providing a principled way to structure polynomial constraints.}
	Section~\ref{sec:roa-theory} establishes conditions under which a set can be certified to be a subset of the ROA for a non-polynomial system and formulates an SOS procedure to find an inner estimate of the ROA.
	Section~\ref{sec:numerical-examples} applies this procedure in two examples.
}

\subsection{Notation}

\(\mathcal{L}_2^n\), \(\mathcal{L}_{2e}^n\), and \(\mathcal{L}_\infty^n\) denote signals \([0, \infty) \to \mathbb{R}^n\) that are square integrable, locally square integrable, and bounded, respectively. \(\mathbb{R}[x]\) denotes polynomials in \(x\).
\(\Sigma[x]\) denotes sum of square polynomials in \(x\): \(p \in \Sigma[x]\) if there exist \(p_1, \dots, p_k \in \mathbb{R}[x]\) such that \(p = \sum_{i=1}^k p_i^2\).

\section{Polynomial Constraints}
\label{sec:poly-constraints}

\new{We consider systems of the form
	\begin{equation}\label{eq:system}
		\begin{aligned}
			\dot{x}(t) & = f(x(t), w(t)) \\
			v(t)       & = g(x(t))       \\
			w(t)       & = \Delta(v)(t)
		\end{aligned}
	\end{equation}
	where $f$ and $g$ are known polynomials, $x(t) \in \mathbb{R}^n$, $f(0, 0) = 0$, $g(0) = 0$, $\Delta(0) = 0$, and $\Delta: \mathcal{L}_{2e}^{n_v} \to \mathcal{L}_{2e}^{n_w}$ is a non-polynomial or uncertain operator.
	This work introduces tools for characterizing~$\Delta$.
}

\begin{definition}
	Let \(\Delta: \mathcal{L}_{2e}^{n_v} \to \mathcal{L}_{2e}^{n_w}\) be a causal operator mapping the signal \(v\) to \(w\) and let \(\Psi\) be a polynomial dynamical system defined as
	\begin{equation}\label{eq:ipc-filter}
		\begin{aligned}
			\dot{x}_\Psi(t) & = f_\Psi(x_\Psi(t), v(t), w(t)), \quad x_\Psi(0) = 0 \\
			z(t)            & = g_\Psi(x_\Psi(t), v(t), w(t))
		\end{aligned}
	\end{equation}
	where \(x_\Psi(t) \in \mathbb{R}^{n_\Psi}\), \(z(t) \in \mathbb{R}\), and \(f_\Psi\) and \(g_\Psi\) are polynomials such that \(f_\Psi(0, 0, 0) = 0\).
	We say \(\Delta\) satisfies the integral polynomial constraint (IPC) defined by \(\Psi\) if
	\begin{equation}\label{eq:ipc-def}
		\int_0^T z(t)dt \geq 0,\quad \forall T\geq 0, v \in \mathcal{L}_{2e}^{n_v}, w = \Delta(v)
	\end{equation}
\end{definition}

Special cases of this include the static polynomial constraint, where \(\Psi(v, w)(t) = p(v(t), w(t))\) for some polynomial \(p\), and the pointwise polynomial constraint, where \(\Psi(v, w)(t) \geq 0\) for all \(t \geq 0\).
The most common case we will discuss is the pointwise static polynomial constraint:
\begin{equation}\label{eq:pointwise-static-poly-constraint}
	\Psi(v, w)(t) = p(v(t), w(t)) \geq 0, \quad \forall t \geq 0
\end{equation}

\new{Integral quadratic constraints (IQCs)~\cite{megretskiSystemAnalysisIntegral1997} are the special case of IPCs with LTI filter dynamics and quadratic filter output.
	Specifically, for the LTI filter
	\begin{equation*}
		\begin{aligned}
			\dot{x}_\Psi(t) & = A_\Psi x_\Psi(t) + B_{\Psi v} v(t) + B_{\Psi w} w(t), \quad x_\Psi(0) = 0 \\
			\tilde{z}(t)    & = C_\Psi x_\Psi(t) + D_{\Psi v} v(t) + D_{\Psi w} w(t)
		\end{aligned}
	\end{equation*}
	with \(A_\Psi\) Hurwitz and \(M = M^\top \in \mathbb{R}^{n_z \times n_z}\), the time-domain hard IQC condition \(\int_0^T \tilde{z}(t)^\top M \tilde{z}(t)\,dt \geq 0\) is recovered by taking \(z = \tilde{z}^\top M \tilde{z}\), a quadratic in \((x_\Psi, v, w)\).}

\section{Constructing Polynomial Constraints}\label{sec:constructing-constraints}

\new{
	We present two methods to construct pointwise static polynomial constraints for non-polynomial functions:
	first using polynomial approximations, and then using nonlinear optimization.
	Section~\ref{sec:transformations} presents results to extend to construction of dynamic IPCs.
}

\new{
	\subsection{Polynomial Approximation} \label{sec:poly-approximation}
}

\new{A Pad\'e approximant $N(x)/D(x)$ approximates $f(x)$ by a ratio of polynomials.
	Given such an approximant, we construct the polynomial constraint}
\begin{equation}\label{eq:pade-approx-constraint}
	\begin{aligned}
		p(x,y) = & (\epsilon_1 x^k D(x) - (yD(x) - N(x)))       \\
		         & \cdot (\epsilon_2 x^k D(x) + (yD(x) - N(x)))
	\end{aligned}
\end{equation}
\new{where $\epsilon_1, \epsilon_2 > 0$ and $k \in \mathbb{N}$ are tunable parameters.
	This construction ensures $p(0, f(0)) = 0$, a useful property in stability analysis.
	A Taylor series truncation may be used in place of $N(x)/D(x)$, though this typically yields looser bounds for the same constraint degree.}

\new{To build intuition of polynomial constraints, we also note that polynomial constraints can be constructed by inspection.}
For the function \(\tanh(x)\), a soft saturation, any local sector bound will be conservative due to the local linear nature and asymptotically saturating nature of the function.
With polynomial constraints, we can consider cubic functions in \(y\) that better approximate \(\tanh\), constructing the polynomial constraint \(p(x, y) = (x - y^3 - y)(\frac{1}{6}y^3 + y - x)\).
A deadzone-style nonlinearity such as \(x-\tanh(x)\) might instead consider a cubic in \(x\), and in general we may rotate polynomial constraints through linear combinations of \(x\) and \(y\).
As another example, consider \(e^x - x - 1\), an exponential with the affine term removed, which we can locally bound with quadratics in \(x\), constructing \(p(x, y) = (y - 0.3x^2)(0.7x^2-y)\).

Figure~\ref{fig:simple-constraints} plots example polynomial constraints, constructed by inspection and by using polynomial approximations.

\begin{figure}[t]
	\centering
	\begin{subfigure}[b]{0.48\textwidth}
		\centering
		\includegraphics[width=\textwidth]{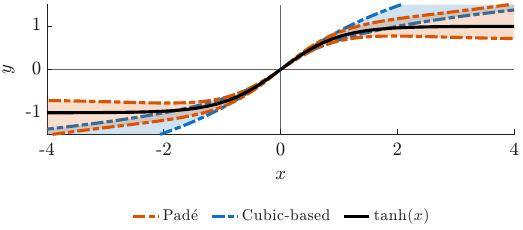}
		\caption{Constraints on $\tanh(x)$. \new{Shaded regions indicate nonnegativity.}}
		\label{fig:tanh}
	\end{subfigure}
	\hfill
	\begin{subfigure}[b]{0.48\textwidth}
		\centering
		\includegraphics[width=\textwidth]{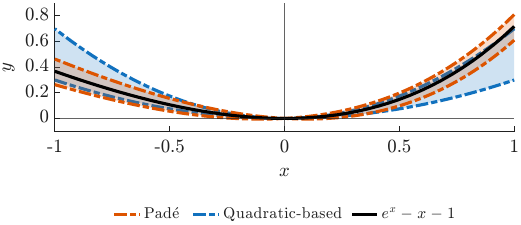}
		\caption{Constraints on $e^x - x - 1$. \new{Shaded regions indicate nonnegativity.}}
		\label{fig:exp}
	\end{subfigure}
	\caption{ 
		Local polynomial constraints on two non-polynomial functions.
		The quadratic-based constraint is a 4th degree polynomial and the cubic-based and Pad\'e approximant-based constraints are 6th degree.
		The Pad\'e approximant-based constraints are both computed using \eqref{eq:pade-approx-constraint} with \(\epsilon_1 = \epsilon_2 = 0.1\), \(k = 1\), and the \([3/2]\) Pad\'e approximant.
	}
	\label{fig:simple-constraints}
\end{figure}

\subsection{Numerical Synthesis} \label{sec:numerical-constraint-synth}

Although simple to construct, approximation-based polynomials lack the tuning flexibility required to incorporate knowledge of the dynamical system into the constraints:
there may exist a constraint that is tighter over the specific local region we are analyzing, or a constraint that is tighter in just the lower or upper bound, only one of which may be of importance given the dynamical system.
To address this, we introduce the following numerical method to construct polynomial constraints, which gives flexibility in tuning the constraint.
This optimization program, which we solve with Matlab's \texttt{fmincon}, may be initialized with constraints constructed via the above-mentioned methods.

We numerically solve for polynomial constraints by solving the optimization program
\begin{mini!}[2]
{p \in \mathbb{R}[x]}{f_0(p)}{\label{eq:synth-direct}}{}
\addConstraint{p(x_{i}, \Delta(x_i))}{\geq 0,}{\quad i=1,\dots,N}{\label{eq:synth-direct-holds}}
\addConstraint{\|\mathrm{coefficients}(p)\|}{=1}{\label{eq:synth-direct-coeff-norm}}
\end{mini!}
where the objective function \(f_0(p)\), described below, is defined to encourage tight constraints, \eqref{eq:synth-direct-holds} ensures \(\Delta\) satisfies the polynomial constraint \(p\) on a constraint evaluation set \(\{x_i\}_{i=1}^{n_c}\), and \eqref{eq:synth-direct-coeff-norm} constrains the scale of the polynomial constraint to ensure the coefficients aren't simply taken to \(0\).
This last constraint leverages that polynomial constraints are invariant under positive scalings, and scaling a polynomial is equivalent to scaling its coefficients.

As a surrogate for tight transformations, we define the objective \(f_0\) to evaluate \(p\) on a set of test points \(\{x_{t,i}, y_{t,i}\}_{i=1}^{n_t}\), penalize points which evaluate nonnegative, and reward points which evaluate negative:
\begin{equation}\label{eq:synth-objective}
	\begin{aligned}
		f(x, y) & = w_o(x, y) \tanh\left(s \cdot p(x, y)\right) \\
		f_0(p)  & = \sum_{i=1}^{n_t} f(x_{t,i}, y_{t,i})
	\end{aligned}
\end{equation}
The parameter \(s > 0\) is a tunable parameter to match the scale of \(\tanh\) to the values of \(p(x, y)\) in the test set, and \(w_o\) is a weighting function: for example, \(w_o(x,y) = 1 + e^{-x^2}\) for tighter constraints near the origin.
We construct a collection of test points with the following process: first, sample points \(\{x_{t,i}\}_{i=1}^{n_{tx}}\) in the domain of \(\Delta\); second, for each \(x_{t,i}\), sample points \(\{y_{t,i,j}\}_{j=1}^{n_{ty}}\) in a neighborhood (of configurable shape) of \(\Delta(x_{t,i})\); and finally, construct test points \(\bigcup_{i=1}^{n_{tx}} \{(x_{t,i}, y_{t,i,j}) \ | \ j=1,\dots,n_{ty}\}\).

While \eqref{eq:synth-direct} ensures construction of a polynomial
\(p\) such that \(\Delta\) satisfies \(p\) and penalizes nonnegative \(p(x, y)\) near the graph of \(\Delta\), it does not constrain or penalize points far from the graph.
Thus, constraints synthesized with this program often require an additional polynomial constraint that \(y\) lie in an interval near the origin, with bounds depending on bounds on \(\Delta\) in the region being verified.
Additionally, because \(p(x, \Delta(x)) \geq 0\) is only ensured at a finite list of evaluation points, we verify the constraint \(p\) after synthesis by using root finders to search for roots of \(p(x, \Delta(x))\) and
ensure there are no roots in the desired local region.

\new{
	\section{Transformations of Polynomial Constraints}\label{sec:transformations}
}
In this section we present results on ensuring the synthesized polynomial constraints are not unnecessarily conservative, and a method for constructing IPCs that leverages the existing body of IQCs.
\begin{theorem} \label{thm:transformation-general}
	Consider \(\Delta\) and \(\tilde{\Delta}\), both operators from \(\mathcal{L}_{2e}^{n_v}\) to \(\mathcal{L}_{2e}^{n_w}\).
	Suppose there exists a polynomial dynamical system \(H: \mathcal{L}_{2}^{n_v} \times \mathcal{L}_{2}^{n_w} \to \mathcal{L}_{2}^{n_v} \times \mathcal{L}_{2}^{n_w} \) such that the restriction of \(H\) to the graph of \(\Delta\),
	\(h = H|_{\mathrm{graph}\ \Delta}\), is an invertible operator from the graph of \(\Delta\) to the graph of \(\tilde{\Delta}\).
	Then \(\tilde{\Delta}\) satisfies the IPC defined by \(\Psi\) if and only if \(\Delta\) satisfies the IPC defined by \(\Psi \circ h\).
\end{theorem}
\begin{proof}
	If \(\tilde{\Delta}\) satisfies \(\Psi\), \(T \geq 0\), and \(v \in \mathcal{L}_{2}^{n_v}\), then \(\int_0^T (\Psi \circ h)(v, \Delta(v))(t)dt = \int_0^T \Psi(\tilde{v}, \tilde{\Delta}(\tilde{v}))(t)dt \geq 0\), so \(\Delta\) satisfies \(\Psi \circ h\).
	If \(\Delta\) satisfies \(\Psi \circ h\), \(T \geq 0\), and \(\tilde{v} \in \mathcal{L}_{2}^{n_v}\), then \(\int_0^T \Psi(\tilde{v}, \tilde{\Delta}(\tilde{v}))(t)dt = \int_0^T (\Psi \circ h \circ h^{-1}) (\tilde{v}, \tilde{\Delta}(\tilde{v}))(t)dt = \int_0^T (\Psi \circ h)(v, \Delta(v))(t)dt \geq 0\), so \(\tilde{\Delta}\) satisfies \(\Psi\).
\end{proof}

Note that \(H\) need not be invertible and \(h^{-1}\) need not be polynomial.
To construct transformations \(h\), the following properties are useful.

\begin{lemma}
	Let \(\Delta\), \(\tilde{\Delta}\), and \(h\) be defined as in Theorem~\ref{thm:transformation-general}.
	Let \(h_1\) and \(h_2\) be the component functions of \(h\) so that \(h(v, \Delta(v)) = (h_1(v, \Delta(v)), h_2(v, \Delta(v))\).
	Then \(h\) is invertible if and only if \(h_1\) is invertible.
\end{lemma}
\begin{proof}
	Define \(g: \tilde{v} \mapsto (\tilde{v}, \tilde{\Delta}(\tilde{v}))\).
	Note that \(g\) is invertible and \(g^{-1}\) is the map projecting onto the first element.
	We can define \(h\) in terms of \(h_1\) by \(h = g \circ h_1\) and \(h_1\) in terms of \(h\) by \(h_1 = g^{-1} \circ h\).
	If \(h\) is invertible, then \(h_1^{-1} = h^{-1} \circ g\).
	If \(h_1\) is invertible, then \(h^{-1} = h_1^{-1} \circ g^{-1}\).
\end{proof}

\begin{corollary} \label{cor:tildeDelta}
	\(\tilde{\Delta} = h_2 \circ h_1^{-1}\).
\end{corollary}

\begin{corollary}
	\(h\) is invertible if and only if \(v \mapsto h_1(v, \Delta(v))\) is invertible.
\end{corollary}

The above results on equivalence of IPCs enables the construction of an IPC on one operator by polynomial transformation of an existing IPC on another operator.
However, showing two operators are related by a polynomial transformation may in general be difficult.
We instead propose starting with an operator \(\Delta\) for which an IPC is desired, and then constructing \(h\) so that the operator \(h_2 \circ h_1^{-1}\), which we define to be \(\tilde{\Delta}\), satisfies a specified class of IQCs.
For example, constructing \(h\) so that \(h_2 \circ h_1^{-1}\) is sector-bounded or slope-restricted by imposing constraints on the coefficients of \(h\) and therefore satisfying sector constraints and Zames-Falb multipliers respectively.
The class of IQCs on \(h_2 \circ h_1^{-1}\) is then transformed into a class of IPCs on \(\Delta\).
This procedure yields
IPCs parameterized by the IQC multipliers, and enables the use of polynomial transformations to create less conservative constraints.
Conservativeness can be lowered further by parameterizing the class of
transformations such that \(h_2 \circ h_1^{-1}\) satisfies the desired IQC instead of picking a single transformation \(h\).

As an illustration, we demonstrate a transformation between \(\tanh\) and a function which satisfies the sector bound \([-1, 1]\) tightly, and use this to construct a constraint on \(\tanh\), recovering the familiar \([0, 1]\) sector bound.
\begin{example} \label{ex:loop}
	Let \(\Delta\) be \(\tanh\)
	and consider the linear transformation \(H \in \R^{2 \times 2}\).
	A sufficient condition for invertibility of \(h\) is invertibility of \(H\).
	Now we use the above procedure to construct a linear transformation \(H\) such that \(\tilde{\Delta}\) satisfies the sector bound \([-1, 1]\).
	A quadratic constraint for this bound is \(M = \mathrm{diag}(1, -1)\).
	Let \(H\) be as follows.
	\[
		H = \begin{bmatrix}
			a & b \\ c & d
		\end{bmatrix}
	\]
	Then \(\Delta\) satisfies \(H^\top M H\) if and only if \((ax + b\tanh(x))^2 - (cx + d\tanh(x))^2 \geq 0\) for all \(x\).
	To select an \(H\) that constructs a tight bound, we use the slope restrictions on \(\tanh\) to show \((ad - bc)/(b^2 - d^2) = 1/2\).
	One solution for this is \(a = 1, b = -1/2, c = 1, d = -3/2\).
	This corresponds to the following \(H\) and polynomial constraint on \(\tanh\):
	\[
		H = \begin{bmatrix}
			1 & -\frac{1}{2} \\ 1 & -\frac{3}{2}
		\end{bmatrix}, \quad H^\top M H = \begin{bmatrix}
			0 & 1 \\ 1 & -2
		\end{bmatrix}.
	\]
	This is the usual sector bound in \([0, 1]\) as expected.
\end{example}

Applying a simple case of this transformation method to \eqref{eq:synth-direct}, we search for \(h\) which transforms a sector constraint into a polynomial constraint on \(\Delta\).
Specifically, we search for an invertible transformation \(h\) that maps the graph of \(\Delta\) to a function that satisfies the quadratic constraint defined by \(\begin{bsmallmatrix} 0 & 0.5 \\ 0.5 & -1 \end{bsmallmatrix}\).
This implies that \(\Delta\) satisfies \(h(x, \Delta(x))^\top M h(x, \Delta(x)) \geq 0\) for all \(x\), which simplifies to \(h_2(x, \Delta(x))\big(h_1(x, \Delta(x)) - h_2(x, \Delta(x))\big) \geq 0\) for all \(x\).
We define \(p(x,y) = h_2(x, y)\big( h_1(x,y) - h_2(x, y) \big)\).
A sufficient condition for invertibility of \(h\) is that the derivative of \(x \mapsto h_1(x, \Delta(x))\) be strictly positive.
Thus, we can numerically solve for a transformation \(h\) by modifying \eqref{eq:synth-direct} to enforce that \(p = h_2 \cdot (h_1 - h_2)\) and add the constraint  \(h_1^\prime(x_i, \Delta(x_i)) > 0\).

\section{Application to Computing Regions of Attraction}\label{sec:roa-theory}

We demonstrate the utility of IPCs by computing inner-estimates of the region of attraction for a non-polynomial system.
We bound the non-polynomial parts of the system with IPCs and apply methods for polynomial systems using SOS programming.

\new{Combining \eqref{eq:system} with the IPC filter \eqref{eq:ipc-filter} yields the extended system with augmented state $x_e = (x, x_\Psi) \in \mathbb{R}^{n + n_\Psi}$:
	\begin{equation}\label{eq:extended-system}
		\begin{aligned}
			\dot{x}_e(t) & = f_e(x_e(t), w(t)) \\
			v(t)         & = g_e(x_e(t))       \\
			z(t)         & = h_e(x_e(t), w(t)) \\
			w(t)         & = \Delta(v)(t)
		\end{aligned}
	\end{equation}}
\begin{theorem}\label{thm:roa}
	\new{Consider the system \eqref{eq:system} with extended system \eqref{eq:extended-system}.}
	Suppose \(\Delta\) satisfies the IPC defined by \(\Psi\), and that \new{\(\Delta\) is bounded input bounded output (BIBO) stable: \(\Delta(v) \in \mathcal{L}_\infty\) whenever \(v \in \mathcal{L}_\infty\)}.
	If there exists a positive definite, continuously differentiable function \(V: \mathbb{R}^{n + n_\Psi} \to \mathbb{R}\), a nonnegative \(s_\Psi \in \mathbb{R}\), and \(\epsilon > 0\) such that \(\{x_e | V(x_e) \leq c\}\) is compact and
	\begin{equation}\label{eq:roa-cond}
		\begin{aligned}
			\nabla & V(x_e)^\top f_e(x_e, w) + s_\Psi h_e(x_e, w) \leq - \epsilon x_e^\top x_e
		\end{aligned}
	\end{equation}
	for all \(w \in \mathbb{R}^{n_w}\) and \(x_e \in \mathbb{R}^{n + n_\Psi}\) such that \(V(x_e) \leq c\), then \(\{x_e | V(x_e) \leq c\} \cap \{x_e = (x, x_\Psi) | x_\Psi = 0\}\) is an inner estimate of the ROA of \eqref{eq:system}.
\end{theorem}
\begin{proof}
	Consider a trajectory of the system such that \(V(x_e(0)) \leq c\).
	\new{Integrating \eqref{eq:roa-cond} from $0$ to $T$ gives
		$V(x_e(T)) - V(x_e(0)) + s_\Psi \int_0^T h_e(x_e, w)\,dt
			\leq -\epsilon \int_0^T x_e^\top x_e\,dt.$
		For the actual trajectory where $w = \Delta(v)$, the IPC gives $\int_0^T h_e(x_e, w)\,dt \geq 0$.
		Since $s_\Psi \geq 0$, it follows that $V(x_e(T)) \leq V(x_e(0)) \leq c$, establishing forward invariance of $\{x_e \mid V(x_e) \leq c\}$.}
	Forward invariance implies \(x_e \in \mathcal{L}_\infty\) by the assumption that the sub-level set is compact.
	\new{Since $g_e$ is continuous and $x_e$ remains in the compact sublevel set, $v = g_e(x_e) \in \mathcal{L}_\infty$, and therefore $w \in \mathcal{L}_\infty$ by the BIBO assumption on $\Delta$.
		Since $f_e$ is continuous and $(x_e, w)$ remains bounded, $\dot{x}_e = f_e(x_e, w) \in \mathcal{L}_\infty$.}
	Additionally, integrating \eqref{eq:roa-cond} shows \(\epsilon \int_0^T x_e(t)^\top x_e(t) dt \leq V(x_e(0))\), so \(x_e \in \mathcal{L}_2\).
	Therefore, \(x_e(t) \to 0\) as \(t \to \infty\) because \(x_e \in \mathcal{L}_2\) and \(\dot{x}_e \in \mathcal{L}_\infty\) by Barbalat's lemma.
	The intersection of \(\{x_e | V(x_e) \leq c\}\) with \(\{x_\Psi = 0\}\) is due to the filter initial condition.
\end{proof}

\begin{remark}\label{rmk:pointwise}
	If \(\Delta\) satisfies \(\Psi\) pointwise, then \(s_\Psi\) can be a nonnegative function of \(x_e\) and \(w\).
\end{remark}

We formulate the ROA estimation problem using SOS programming.
For this section, \(\Psi\) is a pointwise static polynomial constraint, enabling additional SOS multipliers (see Remark~\ref{rmk:pointwise}).
Applying Theorem~\ref{thm:roa}, we seek a positive definite \(V\) and \(s_c, s_\Psi \in \Sigma[(x, w)]\) satisfying the following, where \(\ell_x = \epsilon x^\top x\), \(\ell_{xw} = \epsilon (x^\top x + w^\top w)\), and \(\epsilon \approx 10^{-6}\):
\begin{equation}\label{eq:opt-roa-cond}
	\begin{aligned}
		 & -\nabla V(x)^\top f(x, w) - \ell_{xw}    \\
		 & \hphantom{-} - s_c(x, w)(-V(x) + 1)      \\
		 & \hphantom{-} - s_\Psi(x, w)p(g(x, w), w)
	\end{aligned}
	\in \Sigma[(x, w)]
\end{equation}
To address the bilinearity in \(s_c(x,w)V(x)\), we use a standard alternation~\cite{4283013, iannelli_roa_iqc, chou_synthesizing_2023} between fixing \(V\) and solving for multipliers (expansion step), and fixing multipliers and solving for \(V\) (reshape step).

\begin{algorithm}
	\caption{ROA Estimation via SOS}\label{alg:roa-synth}
	\begin{algorithmic}[1]
		\Require $V_0(x)$ \State $V \gets V_0$
		\While{not \text{converged}}
		\State \textbf{Expand:} fix $V$, max $c$ s.t. \eqref{eq:opt-roa-cond} via bisect. $\to (s_c^*, c^*)$
		\State \textbf{Reshape:} fix $s_c^*$, find $V$ s.t. \eqref{eq:opt-roa-cond}, $c \in [0.9c^*, 0.99c^*]$
		\EndWhile
		\State \Return $\{x \mid V(x) \leq c\}$
	\end{algorithmic}
\end{algorithm}

We use a hyperparameter \(n_V\) for the degree of \(V\) and \(n_\text{total}\) for the degree of each of the polynomials in the constraints in each SOS program.
The degrees of each \(s\)-certificate are determined from \(n_\text{total}, n_V\), and the degrees of any fixed polynomials such as \(p\).
Our implementation\footnote{\url{https://github.com/neelayjunnarkar/polynomial-constraints/}} uses the \textsc{SOSTOOLS}~\cite{sostools} package in Matlab.

\new{\begin{remark}
		If Algorithm~\ref{alg:roa-synth} returns feasible $V$ and $s_\Psi$, then the conditions of Theorem~\ref{thm:roa} are satisfied by construction and the returned sublevel set is an inner estimate of the ROA.
	\end{remark}}

\new{\begin{remark}
		When the algorithm stalls, increase \(n_\text{total}\), then \(n_V\), checking for increase in ROA estimate volume.
		If neither helps, the polynomial constraint may be the bottleneck: evaluate \(p(g(x), \Delta(g(x)))\) along the ROA boundary.
		If \(p \approx 0\) there, the constraint provides no relaxation at the boundary and a tighter or wider-domain constraint is needed.
		Multiple constraints can be used simultaneously via S-procedure, similar to searching in their conic hull.
	\end{remark}}

\new{\begin{remark}
		The SOS programs scale combinatorially with state dimension and \(n_\text{total}\), a standard limitation of SOS-based methods that the IPC framework inherits without worsening.
		For fixed \(n_\text{total}\), replacing a static quadratic IQC with a degree-\(d_p\) polynomial constraint reduces the degree of \(s_\Psi\) by \(d_p - 2\), keeping the SOS program size comparable.
		Further, static pointwise polynomial constraints can achieve tight characterizations of nonlinearities that would otherwise require dynamic IQC filters that increase the state size.
	\end{remark}}

\section{Examples}
\label{sec:numerical-examples}

\subsection{Triple Integrator}

Consider the following triple-integrator system with a saturation on the control input, and a state-feedback controller designed using LQR with weights \(Q=I\) and \(R=1\).
\begin{equation}
	\begin{aligned}
		\dot{x}_1 & = x_2, \quad \dot{x}_2  = x_3                                                               \\
		\dot{x}_3 & = \tanh(u) = \tanh\left(\underbrace{\begin{bsmallmatrix}
				                                                -1 & -2.4142 & -2.4142
			                                                \end{bsmallmatrix}}_{K}x\right)
	\end{aligned}
\end{equation}

We take \(\Delta\) to be the difference between the dynamics and the linearization: \(w = \Delta(v) = \tanh(v) - v\) where \(v = Kx\).
This gives a model in the form of \eqref{eq:system}:
\begin{equation}
	\begin{aligned}
		\dot{x}_1 & = x_2, \quad \dot{x}_2  = x_3             \\
		\dot{x}_3 & = Kx + \Delta(Kx) \\
		\Delta(v) & = \tanh(v) - v.
	\end{aligned}
\end{equation}
We use the method from Section~\ref{sec:constructing-constraints} with weighting function \(w_o(v, w) = 1 + \exp(-w^2)\) to construct a 6th degree local polynomial constraint on \(\Delta\).
The test set of points is designed to penalize high \(|w|\) because the linearization is stable.
This constraint, \(p_\text{num}\), is valid for  \(v \in [-4, 4]\) with the additional constraint that \(w \in [-3, 3]\).
We find benefit in using \(p_\text{num}\) in combination with a \([3/2]\) Pad\'e approximant-based constraint, \(p_\text{pad\'e}\), with \(k=1\), \(\epsilon_1=0.01\), and \(\epsilon_2=0.03\), which is valid for the same \(v\) as \(p_{\text{num}}\).
This constraint is generally looser for \(|w| > |\tanh(v)-v|\) but tighter otherwise, and is incorporated into the SOS programs with its own \(s\)-certificate.
Using a list of constraints is similar to searching over constraints in the conic hull of the list.
The coefficients of these polynomials are available in the associated code repository.
We construct local sector constraints for \(\Delta\) and find the largest ROA with the sector bound valid for \(v \in [-2.1, 2.1]\).
The sector constraint is \( p_\text{sector}(v, w) = -w(w + 0.5379v)\).
The constraints \(p_\text{sector}\), \(p_{\text{num}}\), and \(p_{\text{pad\'e}}\) are plotted in Figure~\ref{fig:triple-int-constraints}.

\begin{figure}[tbp]
	\centering
	\includegraphics[width=\linewidth]{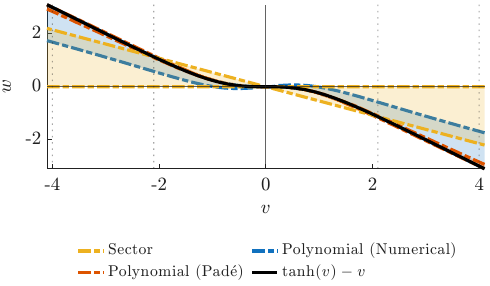}
	\caption{Triple Integrator: Local constraints on \(\tanh(v) - v\). Shaded regions indicate nonnegativity.
		The numerically-constructed polynomial constraint requires the additional constraint that \(w \in [-3, 3]\) for the constraint to be valid for \(v \in [-4, 4]\) to exclude regions where \(|w| \gg |v|\).
		The Pad\'e approximation-based constraint is valid for \(v \in [-4, 4]\), and the sector constraint is valid for \([-2.1, 2.1]\).
		The polynomial constraints give tighter bounds on \(\tanh(v) - v\) than does the sector constraint, while also being valid for a wider region.}
	\label{fig:triple-int-constraints}
\end{figure}

We then use Algorithm~\ref{alg:roa-synth} to synthesize ROAs.
We add additional constraints to the SOS programs to ensure that the verified sub-level sets \(V(x) \leq c\) are subsets of the region where the constraints are valid: \((1 + s_d(x, w))q(x, w) - s_n(x,w)(c-V(x)), s_d(x,w), s_n(x,w) \in \Sigma[(x,w)]\) where \(q(x,w) \geq 0\) represents the region where the constraints are valid.
To compute the ROA estimate using the sector constraint, we compute 50 iterations of Algorithm~\ref{alg:roa-synth} with \(n_V = 2, n_\text{total}=6\), and an initial Lyapunov function from the certificate of the LQR controller.
To compute the ROA estimate using the polynomial constraints, we initialize with the sector-based estimate after 30 iterations, and take 20 additional iterations of \(n_V = 2, n_\text{total}=6\).

\begin{figure}[tbp]
	\centering
	\includegraphics[width=\linewidth]{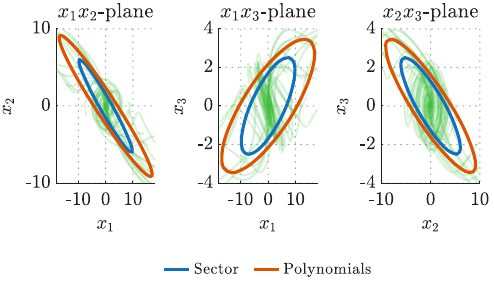}
	\caption{Triple Integrator: Regions of attraction, synthesized using different constraints on \(\Delta\), projected onto planes. The ROA computed using the polynomial constraints has a larger volume.
		Trajectories converging to the origin plotted in green.}
	\label{fig:triple-int-projections}
\end{figure}

Figure~\ref{fig:triple-int-projections} shows the projections of the resulting regions of attraction onto the \(x_1 x_2, x_1 x_3, x_2 x_3\) planes and Table~\ref{tab:triple-int-table} compares the volumes of the ROAs and the computation times (for the polynomial constraints, the time to compute 30 iterations with \(p_\text{sector}\) is included).
All programs were run on an Intel i5-6600K processor.
The ROA computed using the polynomial constraints is \(3.38\) times larger.

\begin{table}[htbp]
	\renewcommand{\arraystretch}{1.3}
	\caption{Triple Integrator}
	\label{tab:triple-int-table}
	\centering
	\begin{tabular}{@{}l S[table-format=3.1] S[table-format=1.2]@{}} 
		\toprule
		\textbf{Constraint Type} & {\textbf{Volume (Relative)}} & {\textbf{Computation Time (s)}} \\
		\midrule
		Sector                   & {153.70  (1.00x)}            & 265                             \\
		Polynomials              & {519.71  (3.38x)}            & 309                             \\
		\bottomrule
	\end{tabular}
\end{table}

\subsection{System with Exponential}

Consider the following system that involves an exponential, which is non-polynomial and asymmetric.
\begin{equation}\label{eq:exp-sys}
	\begin{aligned}
		\dot{x}_1 & = -x_1 - x_2 + e^{x_1} - 1 \\
		\dot{x}_2 & = x_1 - x_2
	\end{aligned}
\end{equation}
In addition to the origin, this system has an equilibrium at approximately \((1.26, 1.26)\).
We examine the ROA of the origin.
As before, we take \(\Delta\) to be the difference between the dynamics and the linearization, so \(\Delta(x_1) = e^{x_1} - x_1 - 1\), and we analyze the ROA of the origin of the following system.
\begin{equation}\label{eq:exp-sys-delta-form}
	\begin{aligned}
		\dot{x}_1   & = -x_2 + \Delta(x_1) \\
		\dot{x}_2   & = x_1 - x_2          \\
		\Delta(x_1) & = e^{x_1} - x_1 - 1
	\end{aligned}
\end{equation}

Using the method from Section~\ref{sec:constructing-constraints} and the objective weighting function \(w_o(v, w) = 1\) if \(v \geq 0\) and \(w_o(v, w) = 0.5\) if \(v < 0\), we construct a 6th degree local polynomial constraint on \(w = \Delta(v)\), which is valid for \(v \in [-4, 2]\) with the additional constraint that \(w \in [0, 3.1]\).
The weighting function is designed to be tight where \(v > 0\).
Coefficients are available in the associated code repository.
We compare with the local sector condition \(p_\text{sector}(v, w) = (0.318v - w)(0.368v + w)\), which is valid for \(v \in [-1, 0.53]\).
The upper-bound of \(0.53\) on \(v\) was the largest bound we found for which the ROA synthesis problem was feasible.
Visualizations in Figure~\ref{fig:exp-sys-constraints}.

\begin{figure}[tbp]
	\centering
	\includegraphics[width=\linewidth]{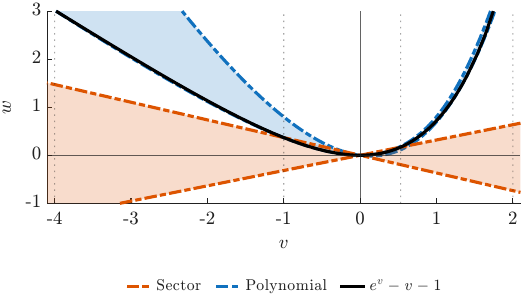}
	\caption{System with Exponential: Local constraints on \(e^v - v - 1\).
		Shaded regions indicate nonnegativity.
		The sector constraint is valid for \(v \in [-1, 0.53]\).
		The polynomial constraint is valid for \(v \in [-4, 2]\) with the additional constraint that \(w \in [0, 3.1]\) to exclude regions where \(|w| \gg |v|\).
		The polynomial constraints give tighter bounds than does the sector constraint, in particular in the upper bound where \(v > 0\), which is important for stability analysis of \eqref{eq:exp-sys}.}
	\label{fig:exp-sys-constraints}
\end{figure}

We synthesize ROAs using Algorithm~\ref{alg:roa-synth}, with an initial quadratic Lyapunov function constructed by solving the Lyapunov equation for the linearized system with \(Q = I\).
When using the sector constraint, we compute 10 iterations with \(n_V = 2, n_\text{total}=6\) followed by 5 iterations with \(n_V = 6, n_\text{total}=10\), beyond which we found no benefit in further iterations.
When using the polynomial constraint, we compute 5 iterations of \(n_V = 2, n_\text{total}=6\), followed by 5 iterations of \(n_V=4, n_\text{total}=8\), and 90 iterations of \(n_V=6, n_\text{total}=10\).
We compare the ROA computed using the polynomial constraint after both 15 total iterations and after 100 total iterations with the ROA computed using the sector constraint.

Figure~\ref{fig:exp-sys-roas} plots the ROAs and Table~\ref{tab:exp-sys-table} compares the volumes and computation times.
The ROAs computed using the polynomial constraint are significantly larger than the one computed using the sector constraint.
In addition, they closely follow the boundary of the region of attraction, unlike the ROA computed using the sector constraint.

\begin{figure}[tbp]
	\centering
	\includegraphics[width=\linewidth]{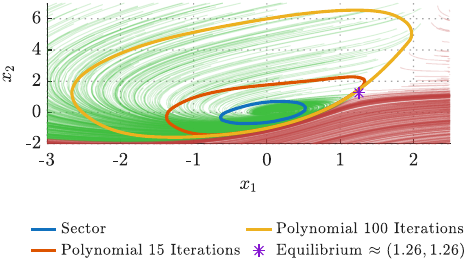}
	\caption{System with Exponential: Inner estimates of ROA of origin, computed using two different characterizations of \(\Delta\).
		The green lines denote trajectories converging to the origin and the red lines denote trajectories diverging from the origin.
		The ROAs computed using the polynomial constraint have larger areas, with additional computation leading to larger ROAs.
	}
	\label{fig:exp-sys-roas}
\end{figure}

\begin{table}[htbp]
	\renewcommand{\arraystretch}{1.3}
	\caption{System with Exponential}
	\label{tab:exp-sys-table}
	\centering
	\begin{tabular}{@{}l S[table-format=3.1] S[table-format=1.2]@{}} 
		\toprule
		\textbf{Constraint Type}   & {\textbf{Volume (Relative)}} & {\textbf{Computation Time (s)}} \\
		\midrule
		Sector                     & {1.16 (1.00x)}               & 67                              \\
		Polynomial, 15 Iterations  & {5.99 (5.16x)}               & 190                             \\
		Polynomial, 100 Iterations & {25.6 (16.0x)}               & 1593                            \\
		\bottomrule
	\end{tabular}
\end{table}

\new{\section{Conclusion}
	This paper introduced polynomial constraints, a generalization of quadratic constraints for characterizing non-polynomial components of dynamical systems, and showed that they integrate naturally with SOS programming.
	Numerical synthesis and transformation methods were presented for constructing polynomial constraints, and two examples demonstrated ROA estimates $3.38$-$16\times$ 
    larger than sector-bound baselines.
	A limitation, shared with IQC-based SOS methods, is that programs scale combinatorially with state dimension.
	Future directions include transforming dynamic IQCs such as Zames-Falb multipliers into  classes of polynomial constraints, 
    and data-driven synthesis of polynomial constraints.
}






\bibliography{refs}


\end{document}